\documentclass{amsart}
\usepackage[a4paper,margin=2.5cm]{geometry}
\usepackage[T1]{fontenc}
\usepackage{amssymb, amsmath, amsthm}
\usepackage{enumitem}
\usepackage{mathtools}
\usepackage{xcolor}
\usepackage[colorlinks=true]{hyperref}
\usepackage{url}

\usepackage[backend=bibtex,style=alphabetic,sorting=nyt,isbn=false,url=false,doi=true,maxalphanames=10,minalphanames=4,mincitenames=4,maxcitenames=10,minnames=4,maxnames=10,giveninits=true,maxbibnames=99]{biblatex}
\AtEveryBibitem{\clearlist{language}}
\renewbibmacro{in:}{}
\theoremstyle{plain}
\newtheorem{theorem}{Theorem}[section]
\newtheorem{conjecture}[theorem]{Conjecture}
\newtheorem{proposition}[theorem]{Proposition}
\newtheorem{lemma}[theorem]{Lemma}
\newtheorem{corollary}[theorem]{Corollary}
\theoremstyle{definition}
\newtheorem{definition}[theorem]{Definition}

\newtheorem{remark}[theorem]{Remark}

\def\oM{\overline{\mathcal{M}}}
\def\sfQ{\mathbf{q}}
\def\sfB{\mathbf{b}}
\def\sfD{\mathbf{d}}

\def\cxdeg{\mathbb{C}\deg}
\newcommand{\Coeff}[1]{\mathop{\mathrm{Coeff}}\nolimits_{\left[#1\right]}}

\title[Examples of the second bi-Hamiltonian structure]{Examples of the second bi-Hamiltonian structure for non-semisimple partial cohomological field theories}

\title[Lifts of CohFTs]{Lifts of partial cohomological field theories and examples of bi-Hamiltonian structures in the non-semisimple case}

\author{Guido Carlet}
\address
{G.~C.: Université Bourgogne Europe, CNRS, IMB UMR 5584, 21000 Dijon, France}
\email{guido.carlet@ube.fr}

\author{Dimitrios Makris}
\address
{D.~M.: Université Bourgogne Europe, CNRS, IMB UMR 5584, 21000 Dijon, France}
\email{dimitrios.makris@ube.fr}

\author{Sergey Shadrin}
\address{S.~S.: Korteweg--de Vriesinstituut voor Wiskunde, Universiteit van Amsterdam, Postbus 94248, 1090GE Amsterdam, Nederland}
\email{s.shadrin@uva.nl}

\begin{document}

\begin{abstract} 
We define the lift of a partial cohomological field theory with respect to a Frobenius algebra and the corresponding lift of local polyvector fields, extending the lift procedure proposed by Della Vedova, Lorenzoni, and Savoldi. This allows us to systematically produce examples of non-semisimple homogeneous partial cohomological field theories whose integrable systems possess a second Hamiltonian structure, thus confirming the conjecture of Buryak et al. on an explicit formula for the second Poisson bracket in new non-semisimple cases. Moreover we present some relations of these lift constructions with the Morimoto theory of the lift of geometric structures to the Weil bundle of infinitely near points associated to a local algebra. 
\end{abstract}

\maketitle

\tableofcontents

\section{Introduction}

\subsection{Dubrovin-Zhang and DR hierarchies} 

The theory of Dubrovin and Zhang associates a bi-Hamiltonian hierarchy of evolutionary PDEs to a semisimple Dubrovin-Frobenius manifold~\cite{dubrovin2001normalformshierarchiesintegrable,LiuWangZhang}, see also~\cite[Introduction]{lorenzoni2026generalisedbihamiltonianstructureshydrodynamic} for a recent survey. It can be further extended to (partial) cohomological field theories at the cost of the second Hamiltonian structure, see~\cite{BPS1,BS22,BLS-DRDZ,BSS25} and a recent survey in~\cite[Preface]{blot2026descendantsintegrableobservablescohomological}. There is an open question whether in the case the underlying partial cohomological field theory is (quasi-)homogenenous (or, in other terminology, conformal), one can still endow the resulting system of flows with the second Hamiltonian structure. Indeed, in the dispersionless limit that corresponds to the genus $0$ part of the corresponding partial cohomological field theory there is a formula due to Dubrovin~\cite{Dub-TFT}, known as Dubrovin's bracket, that doesn't use semisimplicity. 

There is an alternative approach to the construction of a Hamiltonian system of flows associated to partial cohomological field theories, the so-called DR hierarchies, originally invented by Buryak~\cite{Buryak-DR}, and further developed by Buryak and Rossi~\cite{Buryak-Rossi-recursion,BR-spin}, partly in collaboration with Dubrovin and Gu\'er\'e~\cite{BDGR1,BDGR20}. In the theory of DR hierarchies, in the case of a conformal underlying (partial) cohomological field theory, there is a conjectural explicit formula for the Poisson bracket that would supply the second bi-Hamiltonian structure proposed by Buryak et al.~in~\cite{BRS-bracket}. 

The two theories are known to be related by Miura transformations~\cite{BS22,BLS-DRDZ,BSS25}, and this allows to transfer the results from the Dubrovin-Zhang side to DR side and vice versa. This was used, in the semisimple case, by Buryak and Rossi in~\cite{Buryak-Rossi-proof-bracket}: they used the result of~\cite{LiuWangZhang} on the bi-Hamiltonian structure on the Dubrovin-Zhang side in order to prove the conjectural formula~\cite{BRS-bracket} for the second bracket on the DR side in the case of semisimple cohomological field theories. It is important to stress that the methods of~\cite{LiuWangZhang} cannot be extended beyond this case --- this way of reasoning cannot be applied to partial cohomological field theories, and out of reach at the moment in the non-semisimple setting. 

\subsection{New cases of bi-Hamiltonian structure}

The goal of this paper is to provide new cases of established bi-Hamiltonian structures in the non-semisimple setting. The idea comes from the work of Della Vedova, Lorenzoni, and Savoldi~\cite{DVLS}, where they systematically construct non-semisimple examples using the differential-geometric procedure of lift to the tangent bundle due to Yano and Kobayashi~\cite{YanoKoba-1}. We extend their idea in two different ways: we use a more general concept of lift to the bundle of infinitely near points in the sense of Weil, due to Morimoto~\cite{Morimoto-jetbundle,Morimoto-Weil}, and we apply it to the underlying cohomological field theories, in all genera, and construct this way non-semisimple partial cohomological field theories. 

This approach wouldn't give any immediate new results on the Dubrovin-Zhang side of the story. But luckily, on the DR side  the full structure of both Poisson brackets and Hamiltonians, as well as the Hamiltonian reduction, are presented as explicit formulas in terms of the integrals of the underlying partial cohomological field theory.\footnote{Very recently, the same idea of moving the computational complexity to the DR side of the DR/DZ correspondence was also used in~\cite{liu2026finitegroupreductiondrdz} in a different context.}
This allows us to apply the Morimoto lift to all these structures as well and obtain new instances of bi-Hamiltonian structures from the ones proved in the semisimple case in~\cite{Buryak-Rossi-proof-bracket} (thus resolving the conjecture of~\cite{BRS-bracket}, or rather its natural generalization to the case of partial cohomological field theories, in these new cases). Then, of course, one can transfer these results back to the Dubrovin-Zhang side using the Miura equivalence. 

\subsection{Morimoto lift} In the 1960-70s, Morimoto~\cite{Morimoto-Weil} developed a method of prolonging  tensor fields and connections from a manifold $M$ to the so called bundle of infinitely near points of $A$-kind in the sense of Weil. 

Let $A$ be an unital, commutative associative algebra over $\mathbb{R}$, of finite dimension. We call $A$ a local algebra (in the sense of Weil) if it admits a unique maximal ideal $\mathfrak{m}$ such that the quotient $A/\mathfrak{m}$ is one dimensional and such that there exists a non-negative integer $h$ such that $\mathfrak{m}^{h+1} =0$ and $\mathfrak{m}^h \neq0$ .  Important examples of local algebras are finite-dimensional quotients of rings of formal power series in several variables. In fact, every local algebra is isomorphic to such a model.

Let $A$ be a local algebra, $M$ a manifold and $x\in M$. An algebra homomorphism $x':C^\infty(M)\rightarrow A$ is called an $A$-point of $M$ at x if $x'(F) = F(x) \quad \text{mod}\ \mathfrak{m}$. We denote by $M^A_x$ the set of  $A$-points at $x$ and by $M^A=\sqcup_{x\in M} M^A_x$. Moreover, if $\Phi:M \rightarrow N$ is a smooth map, we define $\Phi^A:M^A \rightarrow N^A$ as $\Phi^A(x')(g) = x'(g \circ \Phi)$, for $x'\in M^A_x$ and $g\in C^\infty(N)$. The $A$-lift of a function $F:M \rightarrow \mathbb{R}$ is the  
$A$-valued function $F^A:M^A \rightarrow A\simeq \mathbb{R}^A$. In local coordinates, we have the following formula for the $A$-lift 
\begin{equation} \label{eq:mori-lift}
    F^A = \sum_{m=0}^h \frac{1}{m!}\sum_{2\leq l_1,\dots,l_m \leq L} u^{\alpha_1,l_1}\dots u^{\alpha_m,l_m} \frac{\partial^m F}{\partial u^{\alpha_1}\dots \partial u^{\alpha_m}} f_{l_1}\dots f_{l_m}
\end{equation}
where $\{u^\alpha\}_{\alpha=1,\dots,n}$ are local coordinates on $M$, $\{1=f_1,\dots,f_L\}$ is a basis of $A$ such that $\{f_2,\dots,f_L\}$ is a basis for the maximal ideal $\mathfrak{m}$, and $\{u^{\alpha,l}\}_{\alpha=1,\dots,n;~ l=1,\dots L}$ is the induced system of coordinates on $M^A$, namely, $x'(u^\alpha)=\sum_{l=1}^L u^{\alpha,l} f_l$. (Here and below we always assume summation over repeated indices.)

If $A$ is, in addition, a Frobenius algebra, we can obtain a function $F^c \in C^\infty (M^A)$ that is natural to call the \emph{complete lift}, by composing the $A$-lift with the trace on $A$. Note, however, that this latter concept is not explicitly discussed in the context of infinitely near $A$-points bundle lift in~\cite{Morimoto-Weil}, but rather in the special case of the so-called $(\lambda)$-lift in~\cite{Morimoto-jetbundle}, which we discuss in more detail below. 

 When $A=\mathbb{R}[z]/(z^{2})$, the bundle $M^A$ coincides with the tangent bundle $TM$. In this case~\cite{YanoKoba-1}, the complete lift has a nice intuitive interpretation as the the first jet of the function $f$
 \begin{equation}
     F^c([\gamma] ) =  \frac{d (F\circ \gamma)(\tau)}{d\tau}\bigg|_{\tau =0}
 \end{equation}
where $[\gamma]$ is the germ of a smooth path $\gamma=\gamma(\tau)$ on $M$. This point of view extends to the case of tangent bundles of order $r$ \cite{Morimoto-jetbundle} by taking $A=\mathbb{R}[z] / (z^{r+1})$. The general theory of complete lift on tangent bundles of higher order is well-known and has several applications in differential geometry, such as the theory of affine symmetric spaces and $G$-structures.

\subsection{Applications of the Morimoto lift} We apply below, in the case of Frobenius algebras, a procedure that might be seen as a direct analog of the {complete lift} of Morimoto to partial cohomological field theories and to local polyvector fields, as well as their interaction though the explicit formulas of the theory of DR hierarchies~\cite{Buryak-DR,BRS-bracket}. 
In particular, we generalize the result of~\cite{DVLS} that this lift preserves the Schouten bracket. This allows us to prove the missing parts of the conjectural description of the bi-Hamiltonian structure in~\cite{BRS-bracket} in the cases the lift is applied to the examples when this structure is already established. 

This also leads to an interesting question. According to the theory of Dubrovin-Novikov brackets~\cite{Dubrovin-Novikov}, the ingredients of the brackets are various geometric objects in contravariant differential geometry (see e.~g.~\cite{CC} and references therein)  that go far beyond the cases of tensor fields and connections explicitly studied by Morimoto in~\cite{Morimoto-jetbundle,Morimoto-Weil}. So, one might ask whether the lift that we describe below in the setup of local polyvector fields would match the Morimoto lift on the level of all involved geometric objects. 

\subsection{Acknowledgements}  
G.~C. and D.~M. were supported by the EIPHI Graduate School, contract ``ANR-17-EURE-0002'', grant ``METTHHOD''. S.~S. was supported by the Dutch Research Council, project OCENW.M.21.233.

\section{Partial cohomological field theories and their lifts}

After recalling the definition of homogeneous partial cohomological field theory and of the associated potential, in this section we define the lift of a P-CohFT with respect to a Frobenius algebra and provide a formula for the lifted potential.

\subsection{P-CohFTs} 

Let $V=\langle e_1,\dots,e_N\rangle$ be a finite dimensional vector space over $\mathbb{C}$ with a special element $e_1$, called \emph{unit}, and $\eta\colon V\otimes V\to \mathbb{C}$ a non-degenerate bilinear scalar product, called \emph{metric}. We denote by $\overline{\mathcal{M}}_{g,n}$ the moduli space of stable algebraic curves of genus $g$ with $n$ marked points and by $H^\bullet (\overline{\mathcal{M}}_{g,n},\mathbb{C})$ its even degree cohomology with coefficients in $\mathbb{C}$. 
\begin{definition} \label{def:pcohft}
A \emph{partial cohomological field theory (P-CohFT)} \cite{LiuRuanZhang} on $V$ is a system of linear maps of vector spaces
\begin{align}
    \mathrm{Cl}_{g,n}\colon V^{\otimes n} \to H^\bullet (\overline{\mathcal{M}}_{g,n},\mathbb{C})
\end{align}
for $g,n \geq0$ and $2g-2+n>0$, satisfying the following conditions:
\begin{itemize}
    \item The maps $\mathrm{Cl}_{g,n}$ are $\mathfrak{S}_n$-equivariant with respect to the action on $V^{\otimes n}$ given by permutation of the factors and the action on $H^\bullet (\overline{\mathcal{M}}_{g,n},\mathbb{C})$ induced by the automorphisms of the moduli space associated with  permutation of the labels of the marked points. 
    \item Let $\pi\colon \oM_{g,n+1}\to \oM_{g,n}$ be the map that forgets the last marked point and stabilizes the curve. Then $\pi^*\mathrm{Cl}_{g,n}(\otimes_{i=1}^n e_{\alpha_i}) = \mathrm{Cl}_{g,n+1}(\otimes_{i=1}^n e_{\alpha_i} \otimes e_1)$; moreover, $\mathrm{Cl}_{0,3}(e_\alpha\otimes e_\beta\otimes e_1) = \eta_{\alpha\beta}$. 
    \item Let $\mathrm{gl}\colon \oM_{g_1,n_1+1}\times\oM_{g_2,n_2+1} \to \oM_{g_1+g_2,n_1+n_2}$ be the map that creates a nodal curve out of two curves by identifying into a node their last marked points and identifying the labels on the first and the second component with the ordered subsets $I_1,I_2\subseteq \{1,\dots,n_1+n_2\}$, $|I_1|=n_1$, $|I_2|=n_2$. Then  
        $\mathrm{gl}^*\mathrm{Cl}_{g_1+g_2,n_1+n_2}(\otimes_{i=1}^{n_1+n_2} e_{\alpha_i}) = %\\ \notag & 
        \eta^{\beta\gamma} \mathrm{Cl}_{g_1,n_1+1}(\otimes_{i\in I_1} e_{\alpha_i}\otimes e_\beta) \otimes \mathrm{Cl}_{g_2,n_2+1}(\otimes_{i\in I_2} e_{\alpha_i}\otimes e_\gamma)$.
\end{itemize}
A P-CohFT is called a \emph{cohomological field theory (CohFT)} \cite{KontMan-CohFT} if it satisfies the following extra condition:
\begin{itemize}
    \item Let $\sigma: \oM_{g-1,n+2} \rightarrow \oM_{g,n}$ be the gluing map which assigns a nodal curve to a curve by identifying into a node its two last marked points. Then, $\sigma^* \mathrm{Cl}_{g,n} (\otimes_{i=1}^n e_{\alpha_i}) = \mathrm{Cl}_{g-1,n+2}( \otimes_{i=1}^n e_{\alpha_i} \otimes e_\alpha \otimes e_\beta) \ \eta^{\alpha \beta} $.
\end{itemize}
\end{definition}
\begin{definition} 
The \emph{partition function} of a P-CohFT is defined as $Z = \exp(\epsilon^{-2} F)$, where $F=\sum_{g=0}^\infty \epsilon^{2g} F_g$ is called the \emph{potential} and 
\begin{align} \label{eq:Fg-def}
    F_g \coloneqq \sum_{n=0}^\infty \frac{1}{n!} \int_{\oM_{g,n}} \mathrm{Cl}_{g,n}(\otimes_{i=1}^n e_{\alpha_i}) \sum_{d_1,\dots,d_n=0}^\infty \prod_{i=1}^n \psi_i^{d_i} t^{\alpha_i,d_i}.
\end{align}
Here we assume that for $g=0$ the sum over $n$ starts at $n=3$, and for $g=1$ it starts at $n=1$.
\end{definition}
\begin{definition} \label{def:homogeneity}
	A P-CohFT $\{\mathrm{Cl}_{g,n}\}$ is called \emph{conformal}, or \emph{(quasi-)ho\-mo\-ge\-neous}, if there exist constants ${\sfQ^\alpha}_\beta$ and $\sfB^\alpha$ and $\sfD$, called \emph{homogeneity weights}, such that ${\sfQ^\alpha}_1 = {\delta^\alpha}_1$ (the Kronecker delta) and
	$\eta_{\alpha\mu}{\sfQ^\mu}_\beta + {{\sfQ^{\nu}}_\alpha} \eta_{\nu\beta} = (2-\sfD) \eta_{\alpha\beta}$ and 
	\begin{align} \label{eq: conformal PCohFT}
		& \left( \cxdeg - (g-1) \sfD - n \right) \mathrm{Cl}_{g,n} (\otimes^n_{i=1} e_{\beta_i}) + \sum_{i=1}^{n} { \sfQ^\mu}_{\beta_i} \mathrm{Cl}_{g,n} (\otimes_{j=1}^{i-1} e_{\beta_j} \otimes e_\mu \otimes \otimes_{j=i+1}^n e_{\beta_j}  ) %=0.
		\\ \notag  & \qquad {+ \pi_* \mathrm{Cl}_{g,n+1} (\otimes_{i=1}^n e_{\beta_i} \otimes \sfB^\gamma e_\gamma)} = 0.
	\end{align}
\end{definition}
To simplify some expressions we introduce the following notation: when an upper, resp. lower, index with a bar appears, e.g. $\bar\alpha$, it means that the index has to be lowered, resp. raised, by multiplying the expression by $\eta_{\alpha \bar{\alpha}}$, resp. $\eta^{\alpha \bar{\alpha}}$. In this notation the condition $\eta_{\alpha\mu}{\sfQ^\mu}_\beta + {{\sfQ^{\nu}}_\alpha} \eta_{\nu\beta} = (2-\sfD) \eta_{\alpha\beta}$ above is equivalent to ${\sfQ^\alpha}_\beta = (2-\sfD){\delta^\alpha}_\beta-{\sfQ^{\bar\beta}}_{\bar{\alpha}}$.

In terms of the potential $F$, the homogeneity condition is written
\begin{multline}
	\label{eqn:homogeinity}
	 \left( \sum_{d \geq 0} {(\sfQ  - d)^\mu}_\gamma t^{\gamma,d}\frac{\partial}{\partial t^{\mu,d}} 
	{+ \sfB^\gamma \frac{\partial}{\partial t^{\gamma,0}}
	- \sum_{d \geq 0} \sfB^\beta  C_{\bar\mu\beta\gamma} t^{\gamma,d+1} \frac{\partial}{\partial t^{\mu,d}} }
	+ (3-\sfD) \frac{1}{2} \epsilon \frac{\partial }{\partial \epsilon}  \right) F 
	\\  
	= (3 - \sfD) F
	{+ \frac{1}{2} \sfB^\gamma  C_{\alpha\beta\gamma} t^{\alpha,0} t^{\beta,0} + \epsilon^2 \sfB^\gamma D_\gamma}.
\end{multline}
Here  $C_{\alpha\beta\gamma}\coloneqq 
\int_{\oM_{0,3}} \mathrm{Cl}_{0,3} (e_\alpha\otimes e_\beta \otimes e_\gamma)$ and 
$D_\alpha \coloneqq 
\int_{\oM_{1,1}} \mathrm{Cl}_{1,1} (e_\alpha)$.

\subsection{Lift of P-CohFTs with a Frobenius algebra} \label{cohft-frob}

Let $A = \langle f_1,\dots,f_L \rangle$ be a Frobenius algebra with a unit $f_1$, namely a commutative associative algebra with unit $f_1$, equipped with a map $\int\colon A\to \mathbb{C}$, called \emph{trace}, such that $\omega_{ij} \coloneqq \int f_if_j$ is non-degenerate.

Let $\{\mathrm{Cl}_{g,n}\}$ be a P-CohFT on  $V=\langle e_1,\dots,e_N\rangle$ with the scalar product $\eta$. Define ${}^A\!V\coloneqq V\otimes A = \langle \{ e_{\alpha,i}\coloneqq e_\alpha\otimes f_i\}_{\alpha = 1,\dots,N;\, i = 1,\dots, L} \rangle$ and ${}^A\!\eta(e_{\alpha,i},e_{\beta,j})=  \omega_{ij}\eta_{\alpha\beta}$. Consider a system of maps 
\begin{align}
	{}^A\!\mathrm{Cl}_{g,n}\colon \big({}^A\!V\big)^{\otimes n} \to H^{\bullet}(\oM_{g,n},\mathbb{C})
\end{align}
defined as
\begin{align} \label{eq:b-Cohft-def}
	{}^A\!\mathrm{Cl}_{g,n} (e_{\alpha_1,l_1}\otimes \cdots \otimes e_{\alpha_n,l_n}) \coloneqq \Big(\int f_{l_1}\cdots f_{l_n} \Big) \mathrm{Cl}_{g,n} (e_{\alpha_1}\otimes \cdots \otimes e_{\alpha_n}).
\end{align}
\begin{proposition} \label{prop:Lift-CohFT} 
The maps $\{{}^A\!\mathrm{Cl}_{g,n} \}$ define a P-CohFT on  ${}^A\!V$ with scalar product ${}^A\!\eta$ and unit $e_{1,1}$. Moreover, if the P-CohFT $\{\mathrm{Cl}_{g,n} \}$ is homogeneous with homogeneity weights ${\sfQ^\alpha}_\beta$, $\sfB^\alpha$, and $\sfD$, then $\{	{}^A\!\mathrm{Cl}_{g,n} \}$ is also homogeneous with homogeneity weights  ${{}^A\!\sfQ^{\alpha,l}}_{\beta,j} = {\delta^l}_j {\sfQ^\alpha}_\beta$, ${}^A\!\sfB^{\alpha,l} = {\delta^l}_1 \sfB^{\alpha}$, and ${}^A\!\sfD=\sfD$.
\end{proposition}

\begin{proof} First, we check the axioms of P-CohFT. The $\mathfrak{S}_n$-equivariance is immediate since $\int f_{l_1}\cdots f_{l_n}$ is symmetric in its arguments. The pull-back with respect to $\pi\colon \oM_{g,n+1}\to \oM_{g,n}$ clearly satisfies
\begin{align}
    \pi^*{}^A\!\mathrm{Cl}_{g,n}(\otimes_{i=1}^n e_{\alpha_i,l_i})  &= 
    \Big(\int f_{l_1}\cdots f_{l_n} \Big) \pi^*
    \mathrm{Cl}_{g,n}(\otimes_{i=1}^n e_{\alpha_i})  \\  &= 
    \Big(\int f_{l_1}\cdots f_{l_n} f_1 \Big) 
    \mathrm{Cl}_{g,n+1}(\otimes_{i=1}^n e_{\alpha_i}\otimes e_1) \notag \\
    &= {}^A\!\mathrm{Cl}_{g,n+1}(\otimes_{i=1}^n e_{\alpha_i,l_i}\otimes e_{1,1}). \notag
\end{align}
For the pull-back with respect to $\mathrm{gl}\colon \oM_{g_1,n_1+1}\times\oM_{g_2,n_2+1} \to \oM_{g_1+g_2,n_1+n_2}$ we have
\begin{align}
    \mathrm{gl}^*{}^A&\!\mathrm{Cl}_{g_1+g_2,n_1+n_2}(\otimes_{i=1}^{n_1+n_2} e_{\alpha_i,l_i})  = 
    \Big(\int f_{l_1}\cdots f_{l_{n_1+n_2}} \Big)
    \mathrm{gl}^*\mathrm{Cl}_{g_1+g_2,n_1+n_2}(\otimes_{i=1}^{n_1+n_2} e_{\alpha_i}) 
    \\ \notag & 
    = \Big(\int \prod_{i\in I_1} f_{l_i} \cdot f_{l'} \Big)
    \Big(\int \prod_{i\in I_2} f_{l_i} \cdot f_{l''} \Big)
    \omega^{l'l''} \eta^{\alpha'\alpha''} \mathrm{Cl}_{g_1,n_1+1}(\otimes_{i\in I_1} e_{\alpha_i}\otimes e_{\alpha'}) \mathrm{Cl}_{g_2,n_2+1}(\otimes_{i\in I_2} e_{\alpha_i}\otimes e_{\alpha''})
    \\ \notag & 
    = ({}^A\!\eta)^{\alpha',l',\alpha'',l''}
    {}^A\!\mathrm{Cl}_{g_1,n_1+1}(\otimes_{i\in I_1} e_{\alpha_i,l_i}\otimes e_{\alpha',l'}) {}^A\!\mathrm{Cl}_{g_2,n_2+1}(\otimes_{i\in I_2} e_{\alpha_i,l_i}\otimes e_{\alpha'',l''}).
\end{align}
The first assertion is proved. 

To prove the second assertion let us first check that the homogeneity weights satisfy the conditions of Definition~\ref{def:homogeneity}. Indeed, we have ${{}^A\!\sfQ^{\alpha,l}}_{1,1} = {\delta^l}_1 {\sfQ^\alpha}_1 = {\delta^l}_1{\delta^\alpha}_1 = {\delta^{\alpha,l}}_{1,1}$ and 
\begin{align}
 {}^A\!\eta_{\alpha,l,\mu,i} {{}^A\!\sfQ^{\mu,i}}_{\beta,j} 
+ 
{{}^A\!\sfQ^{\nu,k}}_{\alpha,l} {}^A\!\eta_{\nu,k,\beta,j}
&= \omega_{li}{\delta^i}_j \eta_{\alpha\mu } {\sfQ^{\mu}}_{\beta} 
+ \omega_{kj} {\delta^k}_l {\sfQ^{\nu}}_{\alpha} \eta_{\nu\beta} \\ \notag
&= 
\omega_{lj} (\eta_{\alpha\mu}{\sfQ^\mu}_\beta + {{\sfQ^{\nu}}_\alpha} \eta_{\nu\beta})
\\ \notag 
& = \omega_{lj} (2-\sfD) \eta_{\alpha\beta} = (2-{}^A\!\sfD) {}^A\!\eta_{\alpha,l,\beta,j}.
\end{align}
Then the quasi-homogeneity condition reads: 
\begin{align}
		& \left( \cxdeg - (g-1) {}^A\!\sfD - n \right) {}^A\!\mathrm{Cl}_{g,n} (\otimes^n_{i=1} e_{\beta_i,l_i}) + \sum_{i=1}^{n} { {}^A\!\sfQ^{\mu,k}}_{\beta_i,l_i} {}^A\!\mathrm{Cl}_{g,n} (\otimes_{j=1}^{i-1} e_{\beta_j,l_j} \otimes e_{\mu,k} \otimes \otimes_{j=i+1}^n e_{\beta_j,l_j}  ) %=0.
		\\ \notag  & \qquad + \pi_* {}^A\!\mathrm{Cl}_{g,n+1} (\otimes_{i=1}^n e_{\beta_i,l_i} \otimes {}^A\!\sfB^{\gamma,k} e_{\gamma,k})
        \\ \notag & 
        = \left( \cxdeg - (g-1) \sfD - n \right) \mathrm{Cl}_{g,n} (\otimes^n_{i=1} e_{\beta_i}) \, \Big(\int f_{l_1}\cdots f_{l_n}\Big) 
        \\ \notag & \qquad + \sum_{i=1}^{n} { \sfQ^{\mu}}_{\beta_i} \mathrm{Cl}_{g,n} (\otimes_{j=1}^{i-1} e_{\beta_j} \otimes e_{\mu} \otimes \otimes_{j=i+1}^n e_{\beta_j}  ) \,{\delta^k}_{l_i} \Big(\int f_{l_1}\cdots f_{l_{i-1}} f_k f_{l_{i+1}}\cdots f_{l_{n}}\Big) 
		\\ \notag  & \qquad  + \pi_* \mathrm{Cl}_{g,n+1} (\otimes_{i=1}^n e_{\beta_i} \otimes \sfB^{\gamma} e_{\gamma}) \, {\delta^k}_1 \Big(\int f_{l_1}\cdots f_{l_n} f_k\Big).
\end{align}
The right-hand side of this equation coincides with the left-hand side of~\eqref{eq: conformal PCohFT} multiplied by $\int f_{l_1}\cdots f_{l_n}$, therefore it vanishes. 
\end{proof}

\begin{remark}
This construction might be considered as a certain specialization of the construction of~\cite{KontMan-Product}, especially in the case when the Frobenius algebra $A$ is the purely even cohomology algebra of some even-dimensional variety equipped with the Poincar\'e pairing (we consider only purely even Frobenius algebras as opposed to the $\mathbb{Z}_2$ graded ones, since it allows us to ignore the emerging Koszul signs). 
\end{remark}

\begin{remark}
Note that the lift of a CohFT $\{\mathrm{Cl}_{g,n}\}$ is not necessarily a CohFT. Indeed, the extra condition mentioned in the definition~\eqref{def:pcohft} is satisfied if and only if the following two cohomology classes are equal
\begin{align}
    \sigma^* {}^A\!\mathrm{Cl}_{g,n} (\otimes_{i=1}^n e_{\alpha,l_i})&= \big(\int f_{l_1} \dots f_{l_n}\big) \sigma^* \mathrm{Cl}_{g,n} (\otimes_{i=1}^n e_{\alpha_i}),\\
    {}^A\!\mathrm{Cl}_{g,n} (\otimes_{i=1}^n e_{\alpha_i,l_i} \otimes e_{\alpha, k} \otimes e_{\beta, l}) {}^A\!\eta^{\alpha,k,\beta,l}&= \big(\int f_{l_1}\dots f_{l_n}f_kf_l ~\omega^{kl} \big) \mathrm{Cl}_{g,n} (\otimes_{i=1}^n e_{\alpha_i} \otimes e_{\alpha} \otimes e_\beta) \eta^{\alpha \beta}.
\end{align}
This is the case if and only if the Euler (or handle) element $e:=f_kf_l \omega^{kl}$ is equal to the unit of $A$. In particular this implies that $e$ is invertible and therefore that $A$ is semisimple~\cite{Koch}. % ex.19 p.130
Notice that it also follows that $\int f_1 = \text{dim}(A)$, which does not hold for example in the cases of the cohomology ring of a variety and of the Milnor ring of a singularity, as in those cases $\int f_1$ is zero. 
\end{remark}

\subsection{Lift of the potential} 

Let us denote by $F\{\mathrm{Cl}\}$ the potential of the P-CohFT $\{\mathrm{Cl}_{g,n}\}$ and by 
$F\{{}^A\!\mathrm{Cl}\}$ that of the lifted P-CohFT 
$\{{}^A\!\mathrm{Cl}_{g,n}\}$.
Recall that $F\{\mathrm{Cl}\}$ is a formal power series in the variables $\{t^{\alpha,d}\}$ and $\epsilon$, and $F\{{}^A\!\mathrm{Cl}\}$ is a formal power series in the variables $\{t^{\alpha,l,d}\}$ and $\epsilon$.

For any $l=2,\dots,L$ let $T^l\coloneqq \sum_{\alpha,d} t^{\alpha,l,d} \partial_{t^{\alpha,1,d}}$. 
\begin{corollary} The potentials of a P-CohFT and its lift are related by
\begin{align} \label{eq:lift-of-potential}
    F\{{}^A\!\mathrm{Cl}\} (\{t^{\alpha,l,d}\},\epsilon) = \sum_{m=0}^\infty \frac{1}{m!} \sum_{l_1,\dots,l_m=2}^L \Big(\int f_{l_1}\cdots f_{l_m} \Big) \prod_{i=1}^m T^{l_i} \Big( F\{\mathrm{Cl}\} (\{t^{\alpha,d}\},\epsilon) \big|_{t^{\alpha,d}\to t^{\alpha,1,d}} \Big).
\end{align}
\end{corollary}

\begin{proof}
By definition of the potential of the lifted P-CohFT, Eqs.~\eqref{eq:b-Cohft-def} and~\eqref{eq:Fg-def}, we can write
\begin{equation}
    F\{{}^A\!\mathrm{Cl}\} (\{t^{\alpha,l,d}\},\epsilon) = \int \sum_{g\geq0} \epsilon^{2g} \sum_{n=1}^\infty \frac{1}{n!} \int_{\oM_{g,n}} \mathrm{Cl}_{g,n}(\otimes_{i=1}^n e_{\alpha_i}) \sum_{d_1,\dots,d_n=0}^\infty \sum_{l_1, \dots, l_n=1}^L \prod_{i=1}^n \psi_i^{d_i} t^{\alpha_i,l_i,d_i} f_{l_i},
\end{equation}
which is equal to the trace of original potential $F\{\mathrm{Cl}\} (\{t^{\alpha,d}\},\epsilon)$ in which $t^{\alpha, d}$ is replaced with $\sum_{l=1}^L t^{\alpha, l, d} f_l$. This substitution can be written as the action of $\exp \left( \sum_{l=2}^L t^{\alpha, l, d} f_l \frac{\partial}{\partial t^{\alpha,1,d}} \right)$ on $F\{\mathrm{Cl}\} (\{t^{\alpha,d}\},\epsilon) \big|_{t^{\alpha,d}\to t^{\alpha,1,d}f_1}$. Expanding the exponential and remembering that $f_1$ is the unit of the Frobenius algebra $A$, one obtains~\eqref{eq:lift-of-potential}.
\end{proof}

\begin{remark}
Consider the case of a Frobenius algebra $A$ which is also a local algebra. Then the lift of Morimoto extends to our formal setup, where the coefficients are taken in $\mathbb{C}$ rather than in $\mathbb{R}$ and we work with formal power series in an infinite number of variables rather than with smooth functions on a manifold. 
The $A$-valued lift of a function, see~\eqref{eq:mori-lift}, in this case is given by 
\begin{align} \label{eq:F-A-lift}
{F\{\mathrm{Cl}\}}^A (\{t^{\alpha,l,d}\},\epsilon) = \bigg( \sum_{m=0}^\infty \frac{1}{m!} \Big( \sum_{\alpha,l,d} f_{l} \, t^{\alpha,l,d} \partial_{t^{\alpha,d}} \Big)^m F (\{t^{\alpha,d}\},\epsilon) \bigg) \bigg|_{t^{\alpha,d}\to 0}.
\end{align}
The lift~\eqref{eq:lift-of-potential} of the potential is recovered by applying the trace $\int$ to the $A$-valued lift~\eqref{eq:F-A-lift}, namely we have that
${F}\{\mathrm{Cl}\}^c  = \int {F}\{\mathrm{Cl}\}^A = F\{{}^A\!\mathrm{Cl}\}$.
\end{remark}

\begin{remark}\label{rmk: Milnor-ring}
An important source of Frobenius algebras which are also local algebras are the Milnor local rings of isolated singularities. The Milnor local ring of an isolated singularity $G\colon \mathbb{C}^D \to \mathbb{C}$ given by
\begin{align}
A = \frac{\mathbb{C}[z^1,\dots,z^D]}{(\partial_{z^1}G,\dots,\partial_{z^D}G)}, \qquad \int g(z^1,\dots,z^D) = \frac{1}{(2\pi\mathrm{i})^D}
\int_{C}
g(z^1,\dots,z^D) \frac{dz_1\cdots dz_D}{ \partial_{z^1}G\cdots \partial_{z^D}G},
\end{align}
where the integral is taken over a contour $\partial_{z^i}G = \varepsilon_i$, $i=1,\dots,D$, for some small positive real parameters $\varepsilon_1,\dots,\varepsilon_D$. The unit is the class of $1\in \mathbb{C}[z^1,\dots,z^D]$ in $A$, and $L=\dim A$ is the Milnor number of the singularity $G$~\cite{arnold,Koch}. 

In the case of the local algebra of the $A_2$ singularity, that is, $D=1$ and $G=z^3$, the Milnor number $L$ is equal to $2$, and the algebra is given by $\mathbb{C}[z]/(z^2)$.
As mentioned in the Introduction, in this case the Morimoto lift corresponds to the lift to the tangent bundle described in~\cite{YanoKoba-1}. The lift, in this sense, of the prepotential of the Frobenius manifold (obtained by setting $\epsilon=0$ and $t^{\alpha,d}=t^{\alpha,l,d}=0$ for $d>0$ in~\eqref{eq:lift-of-potential}) was studied in~\cite[Sec.~4.1]{DVLS}.
\end{remark}

\begin{remark}
In general, the lift of a semisimple P-CohFT with respect to a Frobenius algebra $A$ is not semisimple, if $A$ is not semisimple. 
To see this, recall that the structure constants at the origin of the (formal) Frobenius manifold associated to a P-CohFT are given by $C_{\alpha\beta}^\gamma = C_{\alpha\beta\delta}\eta^{\delta\gamma}$, where $C_{\alpha\beta\delta}$ is the three-point correlator introduced just after~\eqref{eqn:homogeinity}. Given a Frobenius algebra $A$ with structure constants $f_kf_l=\Delta_{kl}^r f_r$, it is easy to see that the structure constants at the origin associated with the lift of the CohFT with respect to $A$ are given by $ C_{\alpha,k; \beta,l}^{\gamma,j}= \Delta_{kl}^j\  C_{\alpha \beta}^\gamma$. Therefore, any nilpotent element of $A$ will give rise to nilpotent elements for the Frobenius manifold algebra at the origin. 
\end{remark}

\section{Local polyvector fields and their lifts}

In this section we define the lift of local polyvector fields with respect to a Frobenius algebra and show that it is compatible with the Schouten bracket. Moreover we investigate some relations of such lift with the Morimoto theory of the lift of geometric structures on a manifold to the Weil bundle associated with a local algebra. 

\subsection{Local polyvector fields} 

Let us consider the algebra of densities of local polyvector fields in the variables $u^\alpha$ for $\alpha=1, \dots, N$, namely
\begin{equation}
    \hat{\mathcal{A}} \coloneqq \mathbb{C}[[u^{\alpha,0}]]
    [\{u^{\alpha,s}\}_{s\geq 1},\{\theta^{s}_{\alpha}\}_{s\geq 0}][[\epsilon]].
\end{equation}
It is equipped with the grading $\deg_\theta$ defined on generators as $\deg_\theta(u^{\alpha,s}) = 0$, $\deg_\theta(\theta_\alpha^s) = 1$, $\alpha=1,\dots,N$, $s=0,1,2,\dots$, and $\deg_\theta(\epsilon) = 0$. Let $\hat{\mathcal{A}}^p$ be the homogeneous component of $\hat{\mathcal{A}}$ of  $\deg_\theta$ equal to $p$. When the index $s$ in $u^{\alpha,s}$ is equal to $0$, it is often omitted. 
    The algebra $\hat{\mathcal{A}}$ is also equipped with the standard derivation $\partial_x \coloneqq \sum_{s\geq 0} ( u^{\alpha,s+1}\partial_{u^{\alpha,s}} + \theta_\alpha^{s+1}\partial_{\theta_\alpha^{s}})$ and with the degree $\deg_{\partial_x}$ defined on generators as $\deg_{\partial_x}(u^{\alpha,s}) = s$, $\deg_{\partial_x}(\theta_\alpha^s) = s$, $\alpha=1,\dots,N$, $s=0,1,2,\dots$, and $\deg_{\partial_x}(\epsilon) = -1$. Let $\hat{\mathcal{A}}_d$ be the homogeneous component of $\hat{\mathcal{A}}$ of $\deg_{\partial_x}$ equal to $d$. Let $\hat{\mathcal{A}}^p_d = \hat{\mathcal{A}}^p \cap \hat{\mathcal{A}}_d$.
    The space of local polyvector fields is defined as the quotient
    $\hat{\mathcal{F}} = \hat{\mathcal{A}} / \partial_x \hat{\mathcal{A}}$, with the projection denoted by $\int dx \colon \hat{\mathcal{A}} \to \hat{\mathcal{F}}$. The gradings $\deg_\theta$ and $\deg_{\partial_x}$ descend to $\hat{\mathcal{F}}$, and we use the notations $\hat{\mathcal{F}}^p$ and $\hat{\mathcal{F}}_d$ to denote the respective homogeneous components. 
    Also, let $\hat{\mathcal{F}}^p_d = \hat{\mathcal{F}}^p \cap \hat{\mathcal{F}}_d$.
    
The space $\hat{\mathcal{F}}$ is equipped with the Schouten bracket given on $\overline P = \int dx \, P$ and $\overline Q = \int dx \, Q$, for $P\in \hat{\mathcal{A}}^p$ and  $Q\in \hat{\mathcal{A}}^q$, by
\begin{align}
    [\overline P,\overline Q] \coloneqq \int\!dx\,  \Big( \delta_{u^{\alpha}} P\, \delta_{\theta_{\alpha}} Q + (-1)^{p} \delta_{\theta_{\alpha}} P\, \delta_{u^{\alpha}} Q \Big),
\end{align}
    where $\delta_{u^\alpha},\delta_{\theta_{\alpha}}\colon \hat{\mathcal{A}}\to \hat{\mathcal{A}}$ are the variational derivatives defined as
\begin{align}
    \delta_{u^\alpha} & \coloneqq \sum_{s=0}^\infty (-\partial_x)^s \partial_{u^{\alpha,s}}, & \delta_{\theta_\alpha} & \coloneqq \sum_{s=0}^\infty (-\partial_x)^s \partial_{\theta_{\alpha}^s}.
\end{align}
They both vanish on $\partial_x \hat{\mathcal{A}}$, so they descend to  operators from $\hat{\mathcal{F}}$ to $\hat{\mathcal{A}}$, which we denote by the same symbols.

\subsection{Lift of local polyvector fields with a Frobenius algebra} \label{subsec: Lift with a Frobenius}

As in  \S\ref{cohft-frob}, let $A$ be a Frobenius algebra of dimension $L$. Let us denote by ${}^A\!\hat{\mathcal{A}}$ the algebra of densities of local polyvector fields in the variables $u^{\alpha,l}$ for $\alpha=1, \dots , N$ and $l=1, \dots, L$, namely
\begin{equation}
    {}^A\!\hat{\mathcal{A}} \coloneqq \mathbb{C} [[u^{\alpha,l,0}]]    [\{u^{\alpha,l,s}\}_{s\geq 1},\{\theta^{s}_{\alpha,l}\}_{s\geq 0}][[\epsilon]].
\end{equation}
The corresponding space of local polyvector fields is  ${}^A\!\hat{\mathcal{F}} = {}^A\!\hat{\mathcal{A}} / \partial_x {}^A\!\hat{\mathcal{A}}$ and is equipped with the Schouten bracket 
\begin{align}
    [\overline P,\overline Q] \coloneqq \int\!dx\, \delta_{u^{\alpha,i}} P\, \delta_{\theta_{\alpha,i}} Q + (-1)^{\deg_\theta P} \delta_{\theta_{\alpha,i}} P\, \delta_{u^{\alpha,i}} Q.
\end{align}
It is convenient to introduce an alternative presentation of ${}^A\!\hat{\mathcal{A}}$ and ${}^A\!\hat{\mathcal{F}}$ in terms of generators $\theta_{\alpha}^{i,s} \coloneqq \omega^{il}\theta_{\alpha,l}^{s}$. In these variables the Schouten bracket is given by
\begin{align}
    [\overline P,\overline Q] \coloneqq \int\!dx\, \omega^{ij} \Big( \delta_{u^{\alpha,i}} P\, \delta_{\theta_{\alpha}^j} Q + (-1)^{\deg_\theta P} \delta_{\theta_{\alpha}^i} P\, \delta_{u^{\alpha,j}} Q \Big).
\end{align}

Let us define an algebra homomorphism $\mathsf{I}\colon \hat{\mathcal{A}} \to {}^A\!\hat{\mathcal{A}}$ acting on the generators as
$\mathsf{I}(u^{\alpha,s}) = u^{\alpha,1,s}$ and $\mathsf{I}(\theta_{\alpha}^{s}) = \theta_{\alpha}^{1,s}$. 
This map commutes with $\partial_x$ and thus induces a map $\mathsf{I}\colon \hat{\mathcal{F}} \to {}^A\!\hat{\mathcal{F}}$.

For $l=2,\dots,L$, let us define the derivation $J^l$ of ${}^A\!\hat{\mathcal{A}}$ by
\begin{equation}
    J^l\coloneqq u^{\alpha,l,s}\partial_{u^{\alpha,1,s}} + \theta_{\alpha}^{l,s}\partial_{\theta_{\alpha}^{1,s}}
\end{equation}
which clearly commutes with $\partial_x$. We also define for $i=1,\dots,L$
\begin{align}
   \mathbb{J}_i \coloneqq \int f_i +\sum_{m=1}^\infty \frac{1}{m!} \sum_{2\leq l_1,\dots, l_m \leq L} \Big(\int f_if_{l_1}\cdots f_{l_m}\Big) \prod_{j=1}^m {J^{l_j}}.
\end{align}
\begin{lemma} \label{L-JXY}
    We have $ \mathbb{J}_i(XY)  = \int f_if_jf_k \omega^{jm} \omega^{kl} \mathbb{J}_m(X) \mathbb{J}_l(Y)$
    %$\mathbb{J}_1 (XY) = \omega^{ij} \mathbb{J}_i(X)\mathbb{J}_j(Y)$ 
    for any $X,Y\in {}^A\!\hat{\mathcal{A}}$.
\end{lemma}  
\begin{proof}
    Just rewrite $\mathbb{J}_i = \int f_i \exp\left( \sum_{l=2}^L f_l J^l \right)$ and recall the identity $\omega^{jm} \omega^{kl}\int f_if_jf_k \int f_m a \int f_l b = \int f_i ab$.
\end{proof}

\begin{definition}\label{def:lift-polyvector} The \emph{lift of a $p$-polyvector field} $\overline P$ is defined as ${}^A\!\overline P = \int\!dx \, {}^A\!P$, where ${}^A\!P = \mathbb{J}_1 \mathsf{I} P$.
\end{definition}

\begin{proposition} \label{prop:Schouten commutes with lift} The lift is well defined, that is, it doesn't depend on  the choice of the density $P$ of the polyvector field $\overline P$. Moreover, the lift is compatible with the Schouten bracket, 
namely, for any polyvector fields $\overline P,\overline Q$ we have
    \begin{equation}
        [{}^A\!\overline P, {}^A\!\overline Q] = {}^A\![\overline P,\overline Q].
    \end{equation}
In other words, the lift defines a morphism of the graded Lie algebras $\hat{\mathcal{F}}$ and ${}^A\!\hat{\mathcal{F}}$.
\end{proposition}

\begin{proof} The first assertion follows from the fact that both $\mathbb{J}_1$ and $\mathsf{I}$ commute with $\partial_x$. A simple computation shows that the following identities hold 
\begin{align} \label{eq:delta-I}
    \delta_{u^{\alpha,i}} \mathbb{J}_1 \mathsf{I} = \mathbb{J}_i \delta_{u^{\alpha,1}}\mathsf{I} =  \mathbb{J}_i \mathsf{I} \delta_{u^{\alpha}}\qquad \text{and} \qquad
     \delta_{\theta_{\alpha}^i} \mathbb{J}_1 \mathsf{I} = \mathbb{J}_i \delta_{\theta_{\alpha}^1} \mathsf{I} = \mathbb{J}_i \mathsf{I} \delta_{\theta_{\alpha}},
\end{align}
where the operators $\delta_{u^{\alpha,i}}$, $\delta_{\theta_{\alpha}^i}$ are defined on ${}^A\!\hat{\mathcal{A}}$, while the operators $\delta_{u^{\alpha}}$, $\delta_{\theta_{\alpha}}$ act on $\hat{\mathcal{A}}$.
It follows that
\begin{align} \label{eq:Computation-Schouten-B}
    [{}^A\!\overline P, {}^A\!\overline Q] & = \int\!dx\, \omega^{ij} \Big( \delta_{u^{\alpha,i}} {}^A\!P\, \delta_{\theta_{\alpha}^j} {}^A\!Q + (-1)^{\deg_\theta P} \delta_{\theta_{\alpha}^i} {}^A\! P\, \delta_{u^{\alpha,j}} {}^A\! Q \Big) 
    \\ \notag 
    & = \int\!dx\, \omega^{ij} \Big( \mathbb{J}_i \mathsf{I}  \delta_{u^{\alpha}} P\, \mathbb{J}_j  \mathsf{I} \delta_{\theta_{\alpha}} Q + (-1)^{\deg_\theta P} \mathbb{J}_i \mathsf{I} \delta_{\theta_{\alpha}}  P\, \mathbb{J}_j \mathsf{I} \delta_{u^{\alpha}}  Q \Big)
    \\ \notag & 
    = \int\!dx\, \mathbb{J}_1  \mathsf{I} \Big( \delta_{u^{\alpha}} P\, \delta_{\theta_{\alpha}} Q + (-1)^{\deg_\theta P} \delta_{\theta_{\alpha}} P\, \delta_{u^{\alpha}} Q \Big)
    \\ \notag & 
    = \mathbb{J}_1  \mathsf{I} [\overline P,\overline Q] = {}^A\![\overline P,\overline Q].
\end{align}
\end{proof}

\subsection{Lift of geometric ingredients of Poisson structures} 

According to~\cite{Dubrovin-Novikov}, the density of a hydrodynamic bracket can be written as 
\begin{align}
\frac 12 g^{\alpha\beta} \theta_\alpha \theta_\beta^1 + \Gamma^{\alpha\beta}_\gamma u^{\gamma,1} \theta_\alpha\theta_\beta,     
\end{align}
where, under an extra assumption that $g^{\alpha\beta} = g^{\alpha\beta}(u^{1},\dots,u^{N})$ is invertible, $g_{\alpha\beta}$ is a metric (over the field of complex numbers we omit positive definiteness as a condition) and $\Gamma^{\alpha\beta}_\gamma =\Gamma^{\alpha\beta}_\gamma (u^{1},\dots,u^{N})$ are the Christoffel symbols of the associated contravariant connection, that is, $\Gamma_{\alpha\gamma}^\beta\coloneqq -g_{\alpha\mu}\Gamma^{\mu\beta}_\gamma$ are the standard Christoffel symbols. 

The lift of a hydrodynamic bracket, in the sense of Def.~\ref{def:lift-polyvector}, is again a hydrodynamic bracket, we obtain natural formulas for the lifts of the (inverted) metric and (contravariant) connection. As we mentioned above, it is natural to compare the formulas that we derive from Def.~\ref{def:lift-polyvector} with the complete lift of Morimoto. Since the theory of complete lift is fully developed and has a clear geometric meaning only in the case of the Frobenius algebra of $A_{r+1}$ singularity $\mathbb{R}[z]/(z^{r+1})$~\cite{Morimoto-jetbundle}, we restrict ourselves to this special case.

The latter theory is known under the name of $(\lambda)$-lift. In Sec.~\ref{sec:lambda-lift} we briefly recall the relevant theory following~\cite{Morimoto-jetbundle}, and in Sec.~\ref{sec:lift-comparison} we identify the two constructions for the structures of the hydrodynamic bracket.

\subsubsection{The \texorpdfstring{$(\lambda)$}{(lambda)}-lift to tangent bundles of higher order} \label{sec:lambda-lift}

Let $M$ be a manifold and let $T^rM$ be the tangent bundle of order $r$, namely, the collection of equivalence classes of paths on $M$, where we identify two paths passing from the same base point if their derivatives agree up to order $r$. Note that if $r=1$, this is simply the tangent bundle $TM$. Let $\{u^\alpha\}_{\alpha=1,\dots,n}$ be a set of local coordinates and let $\{u^{\alpha,l} \}_{\alpha=1,\dots,n, ; l=0\dots,r} $ be the induced system of coordinates on $T^rM$, that is, we assign coordinates $u^{\alpha,0}=u^{\alpha}(0), \ u^{\alpha,l}= \frac{1}{l!}\frac{d^l u^{\alpha}}{\partial\tau^l}|_{\tau=0}$, to a path $\gamma=(u^{\alpha}(\tau))$ on $M$.

For $\lambda=0,\dots,r$, the $(\lambda)$-lift of a function $f: M \rightarrow \mathbb{R}$, a vector field $X=X^\alpha \frac{\partial}{\partial u^\alpha}$ and a one form $\omega=\omega_\alpha du^\alpha$ on $M$ are defined by
\begin{align}
    f^{(\lambda)} &= \sum_{m=1}^\lambda \frac{1}{m!} \sum_{\substack{1\leq l_i \leq r \\ l_1+\dots+ l_m=\lambda}} u^{\alpha_1,l_1} \cdots u^{\alpha_m, l_m} \frac{\partial^m f}{\partial u^{\alpha_1}\cdots \partial u^{\alpha_m}}.\label{eq: f-lambda} \\
    X^{(\lambda)} &= \sum_{l = r-\lambda}^r (X^\alpha)^{(l+\lambda -r)} \frac{\partial}{\partial u^{\alpha, l}}, \qquad \omega^{(\lambda)} = \sum_{l=0}^\lambda \omega_\alpha^{(l)} du^{\alpha, \lambda-l}.
    \end{align}
and are uniquely determined by $X^{(\lambda)}(f^{(\mu)}) = (Xf)^{(\lambda+\mu-r)}, \ \omega^{(\lambda)}(X^{(\mu))}) = (\omega(X))^{(\lambda+\mu-r)}$. The $(\lambda)$-lift of any tensor field can be defined recursively, upon requiring compatibility with contractions and total derivatives. In particular, the $(\lambda)$-lift of a contravariant $2$-tensor $g=g^{\alpha\beta} \frac{\partial}{\partial u^{\alpha}} \otimes \frac{\partial}{\partial u^\beta}$ reads
\begin{equation} \label{eq:moto-lift-g}
    g^{(\lambda)} = \sum_{l_1=0}^\lambda \sum_{l_2=0}^{l_1} (g^{\alpha \beta})^{(l_2)} \frac{\partial}{\partial u^{\alpha,r+l_2-l_1}} \otimes \frac{\partial}{\partial u^{\beta, r+l_1-\lambda}}.
\end{equation}
while the $(\lambda)$-lift of a covariant $2$-tensor on $M$, $h=h_{\alpha \beta} du^\alpha \otimes du^\beta$, is given by 
\begin{equation}
    h^{(\lambda)} = \sum_{\mu=0}^\lambda \sum_{\nu =0}^\mu h_{\alpha \beta}^{(\nu)} \ du^{\alpha, \mu-\nu} \otimes du^{\beta, \lambda- \mu}
\end{equation}

\begin{definition}
The complete lift of a geometric object is the $(\lambda)$-lift for $\lambda=r$.    
\end{definition}
%, we shall refer to the $(r)$-lift of a geometric object as the complete lift.

Finally, the complete lift of an affine connection $\nabla\frac{\partial}{\partial u^{\gamma}} = \Gamma^{\alpha}_{\beta \gamma}\frac{\partial}{\partial u^{\alpha}} \otimes du^\beta$ 
is the connection $\nabla^{(r)}$ on $T^rM$ satisfying $\nabla^{(r)}_{X^{(r)}}Y^{(r)} =(\nabla_X Y)^{(r)}$
\begin{equation}\label{eq:moto-lift-nabla}
    \nabla^{(r)} \frac{\partial}{\partial u^{\gamma, l}} = \sum_{l_1=l}^r \sum_{l_2=0}^{l_1-l} \big( \Gamma^{\alpha}_{\beta \gamma})^{(l_1-l-l_2)}\frac{\partial}{\partial u^{\alpha, l_1}} \otimes du^{\beta,l_2}.
\end{equation}
It was also proven in \cite{Morimoto-jetbundle} that the torsion and curvature of $\nabla^{(r)}$ coincide with the $(r)$-lift of the torsion and curvature of $\nabla$. We will need below a formula for the complete lift of the contravariant connection $\tilde{\nabla}^{\alpha\beta}_\gamma = -g^{\alpha\nu} \Gamma_{\nu\gamma}^\beta$, which can be derived from Eqs.~\eqref{eq:moto-lift-g},~\eqref{eq:moto-lift-nabla}:
\begin{equation}\label{eq:moto-lift-contra}
    \tilde\nabla^{(r)} \frac{\partial}{\partial u^{\gamma,l}} = \sum_{l_1=l}^r \sum_{l_2=0}^{l_1-l} \big( \Gamma^{\alpha\beta}_{\gamma})^{(l_1-l-l_2)} \frac{\partial }{\partial u^{\alpha, r-l_2}} \otimes \frac{\partial}{\partial u^{\beta, l_1}}  .
\end{equation}

\subsubsection{Comparison of two constructions}
\label{sec:lift-comparison}

First, we need to match the notation of Morimoto used in Sec.~\ref{sec:lambda-lift} and the notation used in Sec.~\ref{subsec: Lift with a Frobenius}. For a given $r$ the related Frobenius algebra is $A=\mathbb{C}[z]/(z^{r+1})$, of dimension $L=r+1$. For convenience, we change index of the basis of the corresponding Frobenius algebra as $\langle f_0,\dots,f_r\rangle$, with $f_i=[z^i]$, thus $f_0$ being the unit and $\omega_{ij} = \delta_{i+j,r}$. 

After this change of notation, the lift of Def.~\ref{def:lift-polyvector} can be described as follows. We identify $u^{\alpha,0,s}$, resp. $\theta_{\alpha,0}^s$, with $u^{\alpha,s}$, resp. $\theta_{\alpha, r}^s$. Let $J^l\coloneqq \sum_{s} u^{\alpha,l,s}\partial_{u^{\alpha,0,s}} + \theta_{\alpha,r-l}^{s} 
\partial_{\theta_{\alpha,r}^{s}}$. With this preparation and change of notation, we see that according to Def.~\ref{def:lift-polyvector} the $A$-lift of a $p$-polyvector field $\bar{P}$ for $A=\mathbb{C}[z]/(z^{r+1})$ is given by
	\begin{align}
		{}^A\!\bar{P} \coloneqq  \int\!dx \sum_{m=1}^r \frac{1}{m!} \sum_{\substack{1\leq l_i \leq r \\ l_1+\cdots + l_m = r }} \prod_{i=1}^m {J^{l_i}} \sum_{s_1,\dots,s_p} P^{\alpha_1,\dots,\alpha_p}_{s_1,\dots,s_p} \theta_{\alpha_1}^{s_1}\cdots \theta_{\alpha_p}^{s_p}. 
	\end{align}
By immediate substitution, we obtain the following
\begin{corollary} Let $\bar{P} = \int\!dx g^{\alpha\beta} \theta_\alpha \theta_\beta^1 + \Gamma^{\alpha\beta}_{\gamma}u^{\gamma,1}\theta_\alpha\theta_\beta$ and $A=\mathbb{C}[z]/(z^{r+1})$. Then ${}^A\!\bar{P}$ is equal to
\begin{align} \label{eq:r-lift-bracket}
    {}^A\!\bar{P} = \int\!dx \sum_{\substack{0\leq l_1,l_2,l_3 \leq r\\ l_1+l_2+l_3 = r}} (g^{\alpha\beta})^{(l_1)} \theta_{\alpha,r-l_2} \theta_{\beta,r-l_3}^1
    + \sum_{\substack{0\leq l_1,l_2,l_3,l_4 \leq r\\ l_1+l_2+l_3+l_4 = r}} (\Gamma_{\nu\gamma}^\beta)^{(l_1)} u^{\gamma,l_2,1} \theta_{\alpha,r-l_3}\theta_{\beta,r-l_4}. 
\end{align}  
\end{corollary}
Note that up to rearranging the indices the two summands of~\eqref{eq:r-lift-bracket} match precisely the lifts described in Eqs.~\eqref{eq:moto-lift-g},~\eqref{eq:moto-lift-contra}. Thus we obtained the following
\begin{corollary} In the case $A=\mathbb{C}[z]/(z^{r+1})$ the lifts of the inversed metric and contravariant connection prescribed by the Frobenius algebra lift of a non-degenerate hydrodynamic Poisson structure by Def.~\ref{def:lift-polyvector} coincide with complete lifts in the sense of Morimoto, described in Sec.~\ref{sec:lambda-lift}. 
\end{corollary}

This corollary generalizes the $r=1$ case considered in~\cite[Prop.~9]{DVLS}.

\section{Lift of the bi-Hamiltonian structure of the DR hierarchy}

In this section we show that the formula of~\cite{BRS-bracket} provides a bi-Hamiltonian formulation of the DR hierarchy of any \emph{non-semisimple} homogeneous P-CohFT obtained as the lift of a homogeneous semisimple CohFT.

\subsection{Bi-Hamiltonian structure of the DR hierarchy} 

In~\cite{Buryak-DR}, the so-called DR hierarchy was associated with any cohomological field theory, but the construction also works for P-CohFTs, as noted in~\cite{BR-spin}. It is a system of Hamiltonian evolutionary PDEs, where the Hamiltonian structure only depends on the target vector space and the inner product of the P-CohFT, and is given by
\begin{align} \label{eq:Br1}
    \overline{\mathrm{Br}}_{1}\coloneqq \frac{1}{2} \int dx\, \eta^{\alpha\beta} \theta_\alpha^0 \theta_\beta^1 .
\end{align}
It is conjectured in~\cite{BRS-bracket} that in the conformal case the DR hierarchy possesses a bi-Hamiltonian structure (strictly speaking, this conjecture is given in~\emph{op.~cit} for cohomological field theories, but it is natural to extend it to the case of P-CohFTs). In order to recall the conjectural formula proposed in~\cite{BRS-bracket}, we remind the reader that the
construction of the DR hierarchy involves a local functional $\overline{G}$, which, for a P-CohFT $\{\mathrm{Cl}_{g,n}\}$, is defined as
\begin{align}\label{eq:G}
    \overline{G} & \coloneqq \int dx \sum_{\substack{g,n\geq 0 \\ 2g-2+n>0}} \frac{(-\epsilon^2)^g}{n!} \sum_{\substack{d_1,\dots,d_n \geq 0 \\ d_1+\cdots+d_n=2g}} \prod_{i=1}^n u^{\alpha_i,d_i} \times 
    \\ \notag & \qquad \int_{\oM_{g,n}} \mathrm{Cl}_{g,n}(\otimes_{i=1}^n e_{\alpha_i}) \Coeff{\prod_{i=1}^n a_i^{d_i}} \mathbb{D}_g(a_1,\dots,a_n) .
\end{align}
Here $\mathbb{D}_g(a_1,\dots,a_n) \in R^{2g} (\oM_{g,n},\mathbb{C}) \otimes \mathbb{C}[a_1,\dots,a_n]$ is a class whose restriction to any $a_1,\dots,a_n\in \mathbb{Z}$ such that $a_1+\cdots+a_n=0$ coincides with the class $\lambda_g \mathrm{DR}_g(a_1,\dots,a_n)$, where $\lambda_g$ is the top Chern class of the Hodge bundle and $\mathrm{DR}_g(a_1,\dots,a_n)$ is the double ramification cycle. The choice of $\mathbb{D}_g(a_1,\dots,a_n)$ is not unique, and this choice affects the choice of the density $G$ of $\overline{G}$, but only up to an $\partial_x$-exact term (for a discussion of possible choices see e.~g.~\cite[Sec.~1]{rossi2025quantumintegrablehierarchygromovwitten}). 
The Hamiltonians of the DR hierarchy are given by
\begin{align} \label{eq:Ham}
    \overline H_{\alpha,p} & \coloneqq \int dx \sum_{\substack{g,n\geq 0 \\ 2g-1+n>0}} \frac{(-\epsilon^2)^g}{n!} \sum_{\substack{d_1,\dots,d_n \geq 0 \\ d_1+\cdots+d_n=2g}} \prod_{i=1}^n u^{\alpha_i,d_i} \times 
    \\ \notag & \qquad \int_{\oM_{g,n+1}} \psi_1^p \mathrm{Cl}_{g,n+1}(e_\alpha\otimes \otimes_{i=1}^n e_{\alpha_i})   \Coeff{\prod_{i=1}^n a_i^{d_i}} \mathbb{D}_g(0,a_1,\dots,a_n),
\end{align}
for $\alpha=1,\dots,N$ and $p=0,1,2,\dots$.
The flows of the DR hierarchy read
\begin{align}
    \partial_{t^{\beta,p}} \overline F = [[\overline{\mathrm{Br}}_{1},\overline H_{\alpha,p}],\overline F]
\end{align}
for any local functional $F\in \hat{\mathcal{F}}^0$.

Let us assume that the P-CohFT $\{\mathrm{Cl}_{g,n}\}$ is conformal. For any choice of the density $G$ of $\overline{G}$, define the bivector~\cite[Def.~1.12]{BRS-bracket}:
\begin{align} \label{eq:Br2}
	\overline{\mathrm{Br}}_{2} \coloneqq -	\int\! dx\ \bigg(\sum_{p\geq 0} \theta_\alpha^{p+1} \eta^{\alpha\beta} \partial_{u^{ \beta,p}}\bigg) \bigg(\sum_{q\geq 0} \theta_\mu^q {\Big(\sfQ + q -\frac 12 +\frac \sfD 2\Big)^{\mu}}_\alpha \eta^{\alpha\beta} \partial_{u^{\beta,q}} G  {+\frac{1}{2} \theta_\alpha \eta^{\alpha\beta} u^{\gamma} \sfB^\nu  C_{\beta\gamma\nu}} \bigg).
\end{align}

\begin{conjecture} \label{conj:bi-Ham} 
For any homogenenous P-CohFT the bi-vector $\overline{\mathrm{Br}}_{2}$ is Poisson, that is, $[\overline{\mathrm{Br}}_{2},\overline{\mathrm{Br}}_{2}]=0$.
Moreover, we have
\begin{align} \label{eq:bi-Ham-reduction}
    [\overline{\mathrm{Br}}_{2},\overline H_{\alpha,p}] = [\overline{\mathrm{Br}}_{1},\overline H_{\mu,p+1}] {\Big( p+ \frac 52 -\frac\sfD 2  - \sfQ \Big)^\mu}_\alpha + [\overline{\mathrm{Br}}_{1},\overline H_{\mu,p}] \eta^{\mu\beta} \sfB^\gamma C_{\alpha\beta\gamma},
\end{align}
for any $\alpha=1,\dots,N$, $p=0,1,2,\dots$.
\end{conjecture}

\begin{remark} This conjecture is formulated in~\cite[Conj.~1.13]{BRS-bracket} for an arbitrary homogeneous cohomological field theory. But the methods of \emph{op.~cit.} and all subsequent developments in~\cite{BR-spin,BDGR1,BS22,blot2026descendantsintegrableobservablescohomological} work also in the case of P-CohFTs, so it is natural to consider this conjectures in this generalized context.
Moreover, observe that $[\overline{\mathrm{Br}}_{1},\overline{\mathrm{Br}}_{1}]=[\overline{\mathrm{Br}}_{1},\overline{\mathrm{Br}}_{2}]=0$, so if we prove that $\overline{\mathrm{Br}}_{2}$ is Poisson, we obtain a pair of compatible Poisson structures and Eq.~\eqref{eq:bi-Ham-reduction} establishes the bi-Hamiltonian recursion. Notice also that \cite[Conj.~1.13]{BRS-bracket} has been proved for semisimple homogeneous cohomological field theories in~\cite{Buryak-Rossi-proof-bracket}.
\end{remark}

\begin{theorem}\label{thm:lift} Conjecture~\ref{conj:bi-Ham} holds for any P-CohFT obtained as a lift of a homogenenous \emph{semisimple} cohomological field theory.
\end{theorem}

In order to prove this theorem, we first have to check how the lift procedure affects the bi-vector $\overline{\mathrm{Br}}_{2}$ and the Hamiltonians, and then use the result of~\cite{Buryak-Rossi-proof-bracket} that states that Conjecture~\ref{conj:bi-Ham} holds for an arbitrary semisimple homogeneous CohFT.

\subsection{Lift of the brackets} 

Let $\{\mathrm{Cl}_{g,n}\}$ be a partial CohFT and $\{ {}^A\!\mathrm{Cl}_{g,n} \}$ its lift with respect to a Frobenius algebra $A$. Let $\overline{\mathrm{Br}}_1 \{\mathrm{Cl}\}, \overline{\mathrm{Br}}_2 \{\mathrm{Cl}\}$ (respectively, $\overline{\mathrm{Br}}_1 \{{}^A\!\mathrm{Cl}\}, \overline{\mathrm{Br}}_2 \{{}^A\!\mathrm{Cl}\}$) be the bi-vectors associated by Eqs~\eqref{eq:Br1},~\eqref{eq:Br2} to 
$\{\mathrm{Cl}_{g,n}\}$ (respectively, $\{{}^A\!\mathrm{Cl}_{g,n}\}$).

\begin{proposition} \label{prop:Br-commute-with-lift} We have: 
	\begin{align}
		{}^A\!\overline{\mathrm{Br}}_i\{\mathrm{Cl}\} = \overline{\mathrm{Br}}_i\{{}^A\!\mathrm{Cl}\}, \qquad i=1,2.
	\end{align}
\end{proposition}

\begin{proof} For the first bracket, the statement is essentially about the inner product. Indeed, 
\begin{align}
    {}^A\!\overline{\mathrm{Br}}_1\{\mathrm{Cl}\} & = \frac{1}{2} \int dx \, \mathbb{J}_1 \mathsf{I} \big( \eta^{\alpha\beta}\theta_\alpha^0 \theta_\beta^1 \big) = \frac{1}{2} \int dx \, \mathbb{J}_1 \big( \eta^{\alpha\beta}\theta_\alpha^{1,0} \theta_\beta^{1,1} \big)
    = \frac{1}{2} \int dx \, \omega_{jl} \eta^{\alpha\beta}\theta_\alpha^{j,0} \theta_\beta^{l,1} 
    \\ \notag 
    & = \frac{1}{2} \int dx \, \omega^{jl} \eta^{\alpha\beta}\theta_{\alpha,j}^{0} \theta_{\beta,l}^{1} = \frac{1}{2} \int dx \, {}^A\!\eta^{\alpha,j,\beta,l}\theta_{\alpha,j}^{0} \theta_{\beta,l}^{1} = \overline{\mathrm{Br}}_1\{{}^A\!\mathrm{Cl}\}.
\end{align}
For the second bracket, we first notice that by explicit conjugation
\begin{align}
    \mathbb{J}_1 \mathsf{I} \bigg(\sum_{p\geq 0} \theta_\alpha^{p+1} \eta^{\alpha\beta}\partial_{u^{ \beta,p}}\bigg) & = 
    \bigg(\sum_{p\geq 0} \theta_{\alpha,l}^{p+1}\, {}^{A}\!\eta^{\alpha,l,\beta,j} \, \partial_{u^{\beta,j,p}}\bigg) \mathbb{J}_1 \mathsf{I},
    \\ \notag 
    \mathbb{J}_1 \mathsf{I} \bigg(\sum_{q\geq 0} \theta_\mu^q {\Big(\sfQ + q -\frac 12 +\frac \sfD 2\Big)^{\mu}}_\alpha \eta^{\alpha\beta} \partial_{u^{\beta,q}} \bigg) & = \bigg(\sum_{q\geq 0} \theta_{\mu,l}^q {\Big({}^A\!\sfQ + q -\frac 12 +\frac {{}^A\!\sfD} 2\Big)^{\mu,l}}_{\alpha,j}\,{}^A\!\eta^{\alpha,j,\beta,k} \partial_{u^{\beta,k,q}} \bigg) \mathbb{J}_1 \mathsf{I}.
\end{align}
Thus, it is sufficient to check that 
\begin{align} \label{eq:check-G}
    {}^A\!\overline{G}\{\mathrm{Cl}\} = \overline{G}\{{}^A\!\mathrm{Cl}\},
\end{align}
where $\overline{G}\{\mathrm{Cl}\}$ (respectively, $\overline{G} \{{}^A\!\mathrm{Cl}\}$) is the functional associated by Equation~\eqref{eq:G} to $\{\mathrm{Cl}_{g,n}\}$ (respectively, $\{{}^A\!\mathrm{Cl}_{g,n}\}$) and the similar statement for the vector field 
\begin{align}
\overline{X} \coloneqq \int dx \, \theta_\alpha \eta^{\alpha\beta} u^{\gamma} \sfB^\nu  C_{\beta\gamma\nu},    
    \end{align}
where, again, we shall use the notation $\overline{X}\{\mathrm{Cl}\}$ (respectively, $\overline{X} \{{}^A\!\mathrm{Cl}\}$). Notice that the $\overline{\mathrm{Br}}_{2}$ formula given in~\eqref{eq:Br2} is invariant under addition of $\partial_x$-exact terms to either $G$ or $X$. 

First, we compute ${}^A\!\overline{G}\{\mathrm{Cl}\}$. To this end note that 
\begin{align}
    & {\phantom{\bigg|}}^A\!\bigg( \int dx \prod_{i=1}^n u^{\alpha_i,d_i}  \mathrm{Cl}_{g,n}(\otimes_{i=1}^n e_{\alpha_i}) \bigg)
    = \int dx\, \mathbb{J}_1 \prod_{i=1}^n u^{\alpha_i,1,d_i}   \mathrm{Cl}_{g,n}(\otimes_{i=1}^n e_{\alpha_i})
    \\ \notag &
    = \int dx\, \prod_{i=1}^n u^{\alpha_i,l_i,d_i}  \mathrm{Cl}_{g,n}(\otimes_{i=1}^n e_{\alpha_i}) \int f_1 \prod_{i=1}^n f_{l_i} 
    = \int dx\, \prod_{i=1}^n u^{\alpha_i,l_i,d_i}  {}^A\!\mathrm{Cl}_{g,n}(\otimes_{i=1}^n e_{\alpha_i,l_i}). 
\end{align}
This implies Eq.~\eqref{eq:check-G}, since ${}^A\!\overline{G}\{\mathrm{Cl}\}$ (resp., $\overline{G}\{{}^A\!\mathrm{Cl}\}$) is obtained by applying to the left hand side (resp., right hand side) of this identity the operator $\sum_{g,n} \frac{(-\epsilon^2)^g}{n!}\int_{\oM_{g,n}} \Coeff{\prod_{i=1}^n a_i^{d_i}} \mathbb{D}_g(a_1,\dots,a_n)$ and further taking the sum over $d_1,\dots,d_n$. 

Second, let us check the similar statement for the vector $\overline{X}$. Indeed, we have
 \begin{align}\label{eq:computation-X}
    {}^A\!\overline{X}\{\mathrm{Cl}\} &= \int{dx\, \mathbb{J}_1 
    \mathsf{I} \Big( \theta_\alpha^{1} \eta^{\alpha\beta} u^{\gamma, 0} \sfB^\nu C_{\beta\gamma \nu} \Big) } = \int{dx\, \mathbb{J}_1 ( \theta_\alpha^{1,1} u^{\gamma,1,0} \sfB^\nu C_{\beta\gamma \nu} ) }
    \\ \notag 
     &=\int{dx\,  \theta_\alpha^{l,1} u^{\gamma,j,0} \sfB^\nu C_{\beta\gamma \nu} }  \int f_lf_jf_1
     = \int{dx\,  \theta_{\alpha,k}^{1} \omega^{kl} u^{\gamma,j,0}\, {}^A\! \sfB^{\nu,m} \big( C_{\beta\gamma \nu}   \int f_lf_jf_m\big)},
 \end{align}
 where for the last identity we used that ${}^A\!\sfB^{\nu,m} = \sfB^\nu \delta^m_1$. 
 Moreover, observe that 
\begin{align}
C_{\beta\gamma \nu}   \int f_lf_jf_m
= \int_{\oM_{0,3}} \mathrm{Cl}(e_{\beta}\otimes e_{\gamma} \otimes e_{\nu} )\int f_lf_jf_m = \int_{\oM_{0,3}} {}^A\!\mathrm{Cl}(e_{\beta,l}\otimes e_{\gamma,j} \otimes e_{\nu,m}). 
\end{align}
This implies that the right hand side of Eq.~\eqref{eq:computation-X} is indeed equal to $\overline{X} \{{}^A\!\mathrm{Cl}\}$.
\end{proof}

\subsection{Lift of the Hamiltonians} Let $H_{\alpha,p}\{\mathrm{Cl}\}$ (respectively, $H_{\alpha,l,p} \{{}^A\!\mathrm{Cl}\}$) denote the densities of the functionals representing the Hamiltonians associated by Equation~\eqref{eq:Ham} to $\{\mathrm{Cl}_{g,n}\}$ (respectively, $\{{}^A\!\mathrm{Cl}_{g,n}\}$), for some choice of the polynomials $\mathbb{D}_g$. 

\begin{proposition} \label{prop:lift-Hamiltonians}
    We have
    \begin{align}
        \overline{H}_{\alpha,l,p} \{{}^A\!\mathrm{Cl}\} = \int dx \, \mathbb{J}_l \mathsf{I} {H}_{\alpha,p}\{\mathrm{Cl}\}.
    \end{align}
\end{proposition}

\begin{proof} The proof repeats \emph{mutatis mutandis} the proof of Eq.~\eqref{eq:check-G} above. Indeed, using the same argument, it is sufficient to check that 
\begin{align}
& \int dx\, \prod_{i=1}^n u^{\alpha_i,l_i,d_i}  {}^A\!\mathrm{Cl}_{g,n+1}(e_{\alpha,l}\otimes \otimes_{i=1}^n e_{\alpha_i,l_i}) = \int dx\, \prod_{i=1}^n u^{\alpha_i,l_i,d_i}  \mathrm{Cl}_{g,n+1}(e_\alpha \otimes\otimes_{i=1}^n e_{\alpha_i}) \int f_l \prod_{i=1}^n f_{l_i} 
\\ \notag & 
= \int dx\, \mathbb{J}_l \prod_{i=1}^n u^{\alpha_i,1,d_i}  \mathrm{Cl}_{g,n+1}(e_\alpha \otimes \otimes_{i=1}^n e_{\alpha_i})  =  \int dx\, \mathbb{J}_l \mathsf{I} \prod_{i=1}^n u^{\alpha_i,d_i}  \mathrm{Cl}_{g,n+1}(e_\alpha \otimes \otimes_{i=1}^n e_{\alpha_i}).
\end{align}
\end{proof}

\subsection{Proof of Theorem~\ref{thm:lift}} First, we prove that $[\overline{\mathrm{Br}}_2\{{}^A\!\mathrm{Cl}\},\overline{\mathrm{Br}}_2\{{}^A\!\mathrm{Cl}\}]=0$. It is a direct corollary of 
\begin{itemize}
    \item Prop.~\ref{prop:Br-commute-with-lift} that identifies the formula for $\overline{\mathrm{Br}}_2\{{}^A\!\mathrm{Cl}\}$ of the lifted P-CohFT with the lifted $\overline{\mathrm{Br}}_2\{\mathrm{Cl}\}$ of the original semisimple CohFT $\{\mathrm{Cl}_{g,n}\}$;
    \item The result of~\cite[Thm.~3.4]{Buryak-Rossi-proof-bracket} that states that $\overline{\mathrm{Br}}_2\{\mathrm{Cl}\}$ of the original semisimple CohFT is Poisson, that it, it does satisfy $[\overline{\mathrm{Br}}_2\{\mathrm{Cl}\},\overline{\mathrm{Br}}_2\{\mathrm{Cl}\}]=0$; and
    \item Prop.~\ref{prop:Schouten commutes with lift} that states that the Schouten bracket of lifted bi-vector ${}^A\!\overline{\mathrm{Br}}_2\{\mathrm{Cl}\}$ with itself is equal to the lift of the Schouten bracket $\overline{\mathrm{Br}}_2\{\mathrm{Cl}\}$ with itself, thus does vanish. 
\end{itemize}

In order to establish the formula for the bi-Hamiltonian recursion~\eqref{eq:bi-Ham-reduction}, we use Prop.~\ref{prop:Br-commute-with-lift} and Prop.~\ref{prop:lift-Hamiltonians} as well as Prop.~\ref{prop:Lift-CohFT} in order to rewrite, for the lifted P-CohFT $\{{}^A\!\mathrm{Cl}_{g,n}\}$, the left hand side as 
\begin{align} \label{eq:biham-1}
    [\overline{\mathrm{Br}}_{2}\{{}^A\!\mathrm{Cl}\},\overline H_{\alpha,l,p}\{{}^A\!\mathrm{Cl}\}] = [{}^A\!\overline{\mathrm{Br}}_{2}\{\mathrm{Cl}\},\int dx \, \mathbb{J}_l \mathsf{I} {H}_{\alpha,p}\{\mathrm{Cl}\}]
\end{align}
and the two summands on the right hand side as 
\begin{align} \label{eq:biham-2}
    & [\overline{\mathrm{Br}}_{1}\{{}^A\!\mathrm{Cl}\},\overline H_{\mu,j,p+1}\{{}^A\!\mathrm{Cl}\}] {\Big( p+ \frac 52 -\frac{{}^A\!\sfD} 2  - {}^A\!\sfQ \Big)^{\mu,j}}_{\alpha,l} 
    \\ \notag & \qquad \qquad =[{}^A\!\overline{\mathrm{Br}}_{1}\{\mathrm{Cl}\},\int dx \, \mathbb{J}_l \mathsf{I} {H}_{\mu,p+1}\{\mathrm{Cl}\}] {\Big( p+ \frac 52 -\frac{\sfD} 2  - \sfQ \Big)^{\mu}}_{\alpha};
    \\ \label{eq:biham-3}
    & [\overline{\mathrm{Br}}_{1}\{{}^A\!\mathrm{Cl}\},\overline H_{\mu,j,p}\{{}^A\!\mathrm{Cl}\}] \, {}^A\!\eta^{\mu,j,\beta,m}\, {}^A\!\sfB^{\gamma,k} \int_{\oM_{0,3}}{}^A\!\mathrm{Cl}_{0,3}(e_{\beta,m}\otimes e_{\alpha,l} \otimes e_{\gamma,k}) 
    \\ \notag & \qquad \qquad 
    =[{}^A\!\overline{\mathrm{Br}}_{1}\{\mathrm{Cl}\},\int dx \, \mathbb{J}_j \mathsf{I} {H}_{\mu,p}\{\mathrm{Cl}\}]\, \omega^{jm} \eta^{\mu\beta}\sfB^{\gamma} C_{\alpha\beta\gamma} \int f_lf_mf_k \delta^k_1
     \\ \notag & \qquad \qquad 
    =[{}^A\!\overline{\mathrm{Br}}_{1}\{\mathrm{Cl}\},\int dx \, \mathbb{J}_l \mathsf{I} {H}_{\mu,p}\{\mathrm{Cl}\}]\, \eta^{\mu\beta}\sfB^{\gamma} C_{\alpha\beta\gamma} .   
\end{align}
We proceed with the computation of the right hand sides of~\eqref{eq:biham-1},~\eqref{eq:biham-2}, and~\eqref{eq:biham-3} omitting for brevity the explicit dependence $\{\mathrm{Cl}\}$ on the original cohomological field theory in the notation. We use Lemma~\ref{L-JXY} and the following direct generalization of the formulas~\eqref{eq:delta-I} %presented in Sec.~\ref{subsec: Lift with a Frobenius}:
\begin{align}
%    \mathbb{J}_l(XY) & = \int f_lf_if_j \omega^{ik} \omega^{jm} \mathbb{J}_k(X) \mathbb{J}_m(Y);
    \delta_{u^{\alpha,i}} \mathbb{J}_l   \mathsf{I} & = \int f_if_lf_j \cdot \omega^{jk} \mathbb{J}_k \mathsf{I}  \delta_{u^{\alpha}},
&
    \delta_{\theta_{\alpha}^i} \mathbb{J}_l \mathsf{I} =  \int f_if_lf_j \cdot \omega^{jk} \mathbb{J}_k \mathsf{I}  \delta_{\theta_{\alpha}} .
\end{align}
Then we have, for the right hand side of \eqref{eq:biham-1}, we compute in the vein of~\eqref{eq:Computation-Schouten-B}:
\begin{align} \label{eq:biham-1-expanded}
    & [{}^A\!\overline{\mathrm{Br}}_{2},\int dx \, \mathbb{J}_l \mathsf{I} {H}_{\alpha,p}] = [\int dx \mathbb{J}_1 \mathsf{I} {\mathrm{Br}}_{2},\int dx \, \mathbb{J}_l \mathsf{I} {H}_{\alpha,p}] = \int dx\, \omega^{ij} \delta_{\theta_{\beta}^i}\mathbb{J}_1 \mathsf{I} {\mathrm{Br}}_{2} \, \delta_{u^{\beta,j}} \mathbb{J}_l \mathsf{I} {H}_{\alpha,p}
    \\ \notag & 
    = \int dx\, \omega^{ij} \mathbb{J}_i \mathsf{I} \delta_{\theta_{\beta}} {\mathrm{Br}}_{2} \, \mathbb{J}_m \mathsf{I} \delta_{u^{\beta}} {H}_{\alpha,p} \int f_l f_j f_k \omega^{km} = \int dx \, \mathbb{J}_l \mathsf{I} \big( \delta_{\theta_{\beta}} {\mathrm{Br}}_{2} \delta_{u^{\beta}} {H}_{\alpha,p} \big).
\end{align}
In a similar way, the right hand side of \eqref{eq:biham-2} is equal to
\begin{align} \label{eq:biham-2-expanded}
\int dx \, \mathbb{J}_l \mathsf{I} \Big( \delta_{\theta_{\beta}} {\mathrm{Br}}_{1} \delta_{u^{\beta}} {H}_{\mu,p+1} {\Big( p+ \frac 52 -\frac{\sfD} 2  - \sfQ \Big)^{\mu}}_{\alpha} \Big),
\end{align}
and the right hand side of \eqref{eq:biham-3} is equal to
\begin{align} \label{eq:biham-3-expanded}
\int dx \, \mathbb{J}_l \mathsf{I} \big( \delta_{\theta_{\beta}} {\mathrm{Br}}_{1} \delta_{u^{\beta}} {H}_{\mu,p} \eta^{\mu\beta}\sfB^{\gamma} C_{\alpha\beta\gamma} \big).
\end{align}
The result of~\cite[Thm.~3.4]{Buryak-Rossi-proof-bracket} implies that 
\begin{align} 
[\overline{\mathrm{Br}}_{2},\overline H_{\alpha,p}]- [\overline{\mathrm{Br}}_{1},\overline H_{\mu,p+1}] {\Big( p+ \frac 52 -\frac\sfD 2  - \sfQ \Big)^\mu}_\alpha - [\overline{\mathrm{Br}}_{1},\overline H_{\mu,p}] \eta^{\mu\beta} \sfB^\gamma C_{\alpha\beta\gamma} = 0    
\end{align}
for any semisimple homogeneous cohomological field theory, that is, 
\begin{align}
    \delta_{\theta_{\beta}} {\mathrm{Br}}_{2} \delta_{u^{\beta}} {H}_{\alpha,p} - \delta_{\theta_{\beta}} {\mathrm{Br}}_{1} \delta_{u^{\beta}} {H}_{\mu,p+1} {\Big( p+ \frac 52 -\frac{\sfD} 2  - \sfQ \Big)^{\mu}}_{\alpha} - \delta_{\theta_{\beta}} {\mathrm{Br}}_{1} \delta_{u^{\beta}} {H}_{\mu,p} \eta^{\mu\beta}\sfB^{\gamma} C_{\alpha\beta\gamma} \in \partial_x \hat{ \mathcal{A}}
\end{align}
for any $\alpha=1,\dots,N$ $p=0,1,\dots$. Since the map $\mathbb{J}_l \mathsf{I}\colon \hat{ \mathcal{A}}\to {}^A\!\hat{ \mathcal{A}} $ commutes with $\partial_x$, we have
\begin{align}
    \mathbb{J}_l \mathsf{I} \Big( \delta_{\theta_{\beta}} {\mathrm{Br}}_{2} \delta_{u^{\beta}} {H}_{\alpha,p} - \delta_{\theta_{\beta}} {\mathrm{Br}}_{1} \delta_{u^{\beta}} {H}_{\mu,p+1} {\Big( p+ \frac 52 -\frac{\sfD} 2  - \sfQ \Big)^{\mu}}_{\alpha} - \delta_{\theta_{\beta}} {\mathrm{Br}}_{1} \delta_{u^{\beta}} {H}_{\mu,p} \eta^{\mu\beta}\sfB^{\gamma} C_{\alpha\beta\gamma}  \Big) \in \partial_x {}^A\!\hat{ \mathcal{A}}.
\end{align}
This implies that $\eqref{eq:biham-1-expanded}-\eqref{eq:biham-2-expanded}-\eqref{eq:biham-3-expanded}=0$, and hence $\eqref{eq:biham-1}-\eqref{eq:biham-2}-\eqref{eq:biham-3}=0$, which completes the proof of the bi-Hamiltonian recursion~\eqref{eq:bi-Ham-reduction} for any P-CohFT obtained by a lift from a homogenenous semisimple cohomological field theory. 

\printbibliography

\end{document}